\documentclass[a4paper, 11pt]{article}

\usepackage[margin=1.2in]{geometry}
\usepackage{parskip}
\usepackage[titletoc,title]{appendix}

\usepackage{array}
\usepackage{bm}
\usepackage{caption}
\usepackage{enumerate}
\usepackage{rotating}
\usepackage[hidelinks, colorlinks = true, linkcolor = blue, urlcolor = blue, anchorcolor = blue, citecolor = blue]{hyperref}
\usepackage{blkarray, bigstrut}
\usepackage{multirow}
\usepackage{booktabs}
\usepackage{siunitx}
\usepackage{mathtools}
\usepackage{amsmath}
\usepackage{amssymb}
\usepackage{amsthm}
\usepackage{complexity}
\usepackage{graphicx}
\usepackage{pgf}
\usepackage{tikz}
\usetikzlibrary{matrix,patterns,positioning,decorations,arrows,shapes,decorations.markings,calc}
\usepackage{verbatim}
\usepackage{subfiles}
\usepackage{pgfplots}
\usepackage{color}
\usepackage{float}
\usepackage[utf8]{inputenc} 
\usepackage{enumerate}
\usepackage{emptypage}
\usepackage{subcaption}
\usepackage{setspace}
\usepackage[export]{adjustbox}
\usepackage{subfiles}
\usepackage{tabularx, environ}
\newtheorem{theorem}{Theorem}
\newtheorem{definition}[theorem]{Definition}

\newtheorem{lemma}[theorem]{Lemma}
\newtheorem{corollary}[theorem]{Corollary}

\newtheorem{observation}[theorem]{Observation}

\usepackage{xcolor}

\usepackage{subfiles}

\usepackage{subcaption}
\usepackage{tikz}   
\usepackage{amsmath}          
\usetikzlibrary{arrows.meta,decorations.markings,patterns}

\newcommand{\tasks}{\mathcal{I}}
\usepackage{algorithm}
\usepackage{algorithmic}

\title{The Fair Periodic Assignment Problem}

\author{Rolf Nelson van Lieshout$^*$, Bartholomeüs Theodorus Cornelis van Rossum$^*$
\vspace{0.1cm}\\
\small{Operations, Planning, Accounting \& Control}\\ 
\small{Eindhoven University of Technology, The Netherlands}
\vspace{0.1cm}\\
\small{$^\star$Corresponding author} \\
\vspace{0.1cm} \\
\small{r.n.v.lieshout@tue.nl, b.t.c.v.rossum@tue.nl}
}

\date{}

\begin{document}
  
% Title
\maketitle

% Abstract
\begin{abstract}
\noindent 
We study the periodic assignment problem, in which a set of periodically repeating tasks must be assigned to workers within a repeating schedule. The classical efficiency objective is to minimize the number of workers required to operate the schedule. We propose a $\mathcal{O}(n \log n)$ algorithm to solve this problem. Next, we formalize a notion of fairness among workers, and impose that each worker performs the same work over time. We analyze the resulting trade-off between efficiency and fairness, showing that the price of fairness is at most one extra worker, and that such a fair solution can always be found using the Nearest Neighbor heuristic. We characterize all instances that admit a solution that is both fair and efficient, and use this result to develop a $\mathcal{O}(n \log n)$ exact algorithm for the fair periodic assignment problem. Finally, we show that allowing aperiodic schedules never reduces the price of fairness.
\vspace{5mm}
\end{abstract}

\section{Introduction}

Public transport schedules exhibit a high degree of periodicity across multiple time scales. On short time scales, timetables often repeat every hour or even more frequently; on longer time scales, crew rosters typically follow multi-week cycles. This recurring structure motivates the study of the \emph{Periodic Assignment Problem} (PAP), where resources—such as vehicles, platforms, or crew members—must be assigned to tasks within a repeating schedule \cite{bortoletto_et_al,breugem2022equality,vanlieshout2021integrated}.

When assigning vehicles or platforms, the primary objective is operational efficiency and adherence to capacity constraints. However, when assigning crew members, fairness becomes a central concern: it is desirable that all employees perform the same work over time, rather than being locked into fixed subsets of tasks. While fair periodic rosters are widely used in both rail \cite{breugem2022equality} and bus systems \cite{xie2015cyclic}, the fundamental trade-off between fairness (in workload distribution) and efficiency (in the number of required workers) has not been formally quantified.

In this paper, we make this trade-off explicit. We formalize a natural fairness criterion—requiring that all workers cyclically perform the same set of tasks—and study its impact on scheduling efficiency. We characterize the structure of fair periodic assignments and present efficient algorithms to compute fair schedules that require a minimal number of workers.

\subsection{Problem Description}

Consider a set of $n$ timed tasks $\mathcal{I}$ to be performed periodically by a pool of homogeneous workers. The schedule has a fixed period $T$ (typically one week in rostering contexts), and each task $i \in \mathcal{I}$ is defined by a $T$-periodic open interval $(a_i, b_i)$ with $a_i,b_i\in [0,T)$. This interval may wrap around the end of the period, i.e., it is possible that $a_i>b_i$. The duration of task $i$ is 
\[
c(i) := [b_i - a_i]_T,
\]
where $[\cdot]_T$ denotes the modulo-$T$ operator mapping values to $[0, T)$.

The $r$-th occurrence of task $i$ is performed in the interval $(a_i + rT,\ b_i + rT)$. Workers are immediately available for a new task after completing one (any required rest time can be absorbed into task durations). A worker who completes the $r$-th occurrence of task $i$ can begin the $r$-th occurrence of task $j$ if $b_i \le a_j$, or the $(r+1)$-th occurrence of task $j$ otherwise. In general, the \emph{transition time} between task $i$ and task $j$ is independent of $r$ and defined as
\[
c_{ij} := [a_j - b_i]_T.
\]
 These transitions define a complete directed transition graph $\mathcal{G} = (\mathcal{I}, \mathcal{A})$, where $\mathcal{A} := \mathcal{I} \times \mathcal{I}$ includes all possible task-to-task transitions.

We now formalize the optimization problem of finding an efficient periodic assignment.

\begin{definition}\label{def:pap}
Given a period $T \in \mathbb{N}$, a set of $T$-periodic tasks $\mathcal{I}$, and a transition graph $\mathcal{G} = (\mathcal{I}, \mathcal{A})$, the \emph{Periodic Assignment Problem (PAP)} is to find a subset of arcs $\mathcal{M} \subseteq \mathcal{A}$ such that:
\begin{enumerate}[a)]
    \item The total transition time $\sum_{(i,j) \in \mathcal{M}} c_{ij}$ is minimized,
    \item For every task $i \in \mathcal{I}$, $\mathcal{M}$ includes exactly one arc entering and one arc leaving $i$.

\end{enumerate}
\end{definition}

The second condition requires every task to have exactly one predecessor and one successor. Consequently, any feasible solution $\mathcal{M}$ defines a collection of disjoint cycles, where each cycle represents the sequence of tasks repeatedly executed by a group of workers. A cycle $\mathcal{C}$ covering tasks $\tasks_\mathcal{C}$ with transitions $\mathcal{A}_\mathcal{C}$ requires $w(\mathcal{C})$ workers, where 
$$w(\mathcal{C}) = \frac{1}{T}\left(\sum_{i\in \tasks_\mathcal{C}} c(i) + \sum_{(i,j)\in \mathcal{A}_\mathcal{C}} c_{ij}\right).$$ 
Since the durations of tasks are fixed, minimizing total transition time is equivalent to minimizing the number of workers required to operate the schedule.

\autoref{fig:solutions_efficient} shows a PAP instance, together with three possible representations of an efficient solution that requires the minimum of two workers. In \autoref{subfig:arc_efficient}, the tasks to be performed periodically are labeled $I_1, \dots, I_4$ and displayed as circular arcs. Dotted arcs indicate the transition arcs in the efficient solution, requiring two workers to operate two disjoint schedules of two tasks each. \autoref{subfig:graph_efficient} represents the same solution as a set of disjoint directed cycles in the transition graph, where each node corresponds to a task and the durations of the transitions are given on the arcs. Finally, \autoref{subfig:schedule_efficient} visualizes the solution explicitly as a periodic schedule for two workers.

\newcommand{\arc}[4]{
    % #1: x-coordinate, #2: y-coordinate, #3: radius
    \pgfmathsetmacro{\innerR}{#3-0.1}
    \pgfmathsetmacro{\outerR}{#3+0.1}
    \pgfmathsetmacro{\middle}{0.1 * #1 + 0.9 * #2}
    \draw[black] ({#3*cos(#1)}, {#3*sin(#1)}) arc (#1:#2:#3);
    \draw[black] ({\innerR *cos(#1)}, {\innerR * sin(#1)}) -- ({\outerR*cos(#1)}, {\outerR*sin(#1)});
    \draw[black] ({\innerR *cos(#2)}, {\innerR*sin(#2)}) -- ({\outerR*cos(#2)}, {\outerR*sin(#2)});
    \node at ({(\outerR + 0.11) *cos(\middle)}, {(\outerR + 0.11) *sin(\middle)}) {#4};
}

\newcommand{\transitionarc}[4]{
    % #1: start angle, #2: end angle, #3: start radius, #4: end radius
    \pgfmathsetmacro{\steps}{30}
    \pgfmathsetmacro{\deltaAngle}{(#2 - #1)/\steps}
    \pgfmathsetmacro{\deltaRadius}{(#4 - #3)/\steps}
    \foreach \i in {0,...,\steps} {
        \pgfmathsetmacro{\angleA}{#1 + \i * \deltaAngle}
        \pgfmathsetmacro{\angleB}{#1 + (\i + 1) * \deltaAngle}
        \pgfmathsetmacro{\radiusA}{#3 + \i * \deltaRadius}
        \pgfmathsetmacro{\radiusB}{#3 + (\i + 1) * \deltaRadius}
        \draw[loosely dotted] 
            ({\radiusA*cos(\angleA)}, {\radiusA*sin(\angleA)}) -- 
            ({\radiusB*cos(\angleB)}, {\radiusB*sin(\angleB)});
    }

   % Arrowhead at the START point, pointing reverse (counterclockwise tangent)
    \pgfmathsetmacro{\startAngle}{#1}
    \pgfmathsetmacro{\startRadius}{#3}

    % Small step along the arc forward (clockwise)
    \pgfmathsetmacro{\smallStepAngle}{\startAngle + 5}  % 5 degrees ahead (adjust if needed)
    \pgfmathsetmacro{\smallStepRadius}{\startRadius + 0.02} % slightly larger radius for smoothness

    % Position where arrow starts (slightly forward on the arc)
    \coordinate (arrowstart) at ({\smallStepRadius*cos(\smallStepAngle)}, {\smallStepRadius*sin(\smallStepAngle)});

    % Vector direction: from arrowstart back to start point (to get reversed arrow)
    \coordinate (arrowtip) at ({\startRadius*cos(\startAngle)}, {\startRadius*sin(\startAngle)});

    \draw[dotted, ->] (arrowstart) -- (arrowtip);
}

\begin{figure}[htbp!]
\begin{subfigure}[b]{0.3\textwidth}
\centering
\begin{tikzpicture}
    % Draw the full circle.
    \draw[black] (0, 0) circle(1);

    % Draw arcs.
    \arc{15}{165}{1.5}{$I_1$};
    \arc{195}{345}{1.5}{$I_2$};
    \arc{105}{255}{2}{$I_3$};
    \arc{285}{435}{2}{$I_4$};

    % Draw transition arcs.
    \draw[dotted, <-] ({1.5*cos(165)}, {1.5*sin(165)}) arc (165:195:1.5);
    \draw[dotted, <-] ({1.5*cos(345)}, {1.5*sin(345)}) arc (345:375:1.5);
    \draw[dotted, <-] ({2*cos(255)}, {2*sin(255)}) arc (255:285:2);
    \draw[dotted, <-] ({2*cos(435)}, {2*sin(435)}) arc (75:105:2);

    % Draw time markers.
    \foreach \angle/\label in {
    90/$0$,
    0/{$\frac{T}{4}$},
    270/{$\frac{T}{2}$},
    180/{$\frac{3T}{4}$}
    } 
    {
    \draw[thick] ({0.9*cos(\angle)}, {0.9*sin(\angle)}) -- ({1.1*cos(\angle)}, {1.1*sin(\angle)});
    %\node at ({1.1*cos(\angle)}, {1.1*sin(\angle)}) {\label};
    }
\end{tikzpicture}
\caption{Arc representation}
\label{subfig:arc_efficient}
\end{subfigure}
\hfill
\begin{subfigure}[b]{0.3\textwidth}
\centering
\begin{tikzpicture}
    % Tasks
    \node[draw=black, shape=circle] at (0, 2.5) (1) {$I_1$};
    \node[draw=black, shape=circle] at (3, 2.5) (2) {$I_2$};
    \node[draw=black, shape=circle] at (0, 0) (3) {$I_3$};
    \node[draw=black, shape=circle] at (3, 0) (4) {$I_4$};
    
    % Transitions
    \draw[->, bend right] (1) to node[midway, above] {$\frac{T}{12}$} (2);
    \draw[->, bend right] (2) to node[midway, above] {$\frac{T}{12}$} (1);
    \draw[->, bend right] (3) to node[midway, above] {$\frac{T}{12}$} (4);
    \draw[->, bend right] (4) to node[midway, above] {$\frac{T}{12}$} (3);
\end{tikzpicture}
\caption{Graph representation}
\label{subfig:graph_efficient}
\end{subfigure}
\hfill
\begin{subfigure}[b]{0.3\textwidth}
\centering
\resizebox{\linewidth}{!}{
\begin{tikzpicture}
    % Draw payoff axis.
    \draw[->, thin] (0,0.5)--(2.33, 0.5);
    \draw[->, thin] (0,0.5)--(0, 2.5) node[above]{Worker};
    
    % Draw axis labels.
    \draw[thin] (0, 0.4) node[below]{$0$}--(0, 0.6);
    \draw[thin] (1, 0.4) node[below]{$\frac{T}{2}$}--(1, 0.6);
    \draw[thin] (2, 0.4) node[below]{$T$}--(2, 0.6);

    \draw[thin] (-0.1, 1) node[left]{$1$}--(0.1, 1);
    \draw[thin] (-0.1, 2) node[left]{$2$}--(0.1, 2);

    % Draw tasks worker 1.
    \node[draw, thin, dotted, minimum width=0.8333cm, minimum height=0.5cm, align=center, pattern=north west lines] at (0.5, 1) {\scriptsize{$I_4$}};
    \node[draw, thin, dotted, minimum width=0.8333cm, minimum height=0.5cm, align=center, fill=black!40] at (1.5, 1) {\scriptsize{$I_3$}};

    % Draw tasks worker 2.
    \node[draw, thin, dotted, minimum width=0.4166cm, minimum height=0.5cm, align=center] at (0.2033, 2) {};
    \node[draw, thin, dotted, minimum width=0.8333cm, minimum height=0.5cm, align=center, fill=black!15] at (1, 2) {\scriptsize{$I_2$}};
   % \node[draw, thin, dotted, minimum width=0.4166cm, minimum height=0.5cm, align=center] at (1.8033, 2) {\scriptsize{$I_1$}};

    \def\x{1.7933}
    \def\y{2}
    \def\w{0.41333}
    \def\h{0.5}

    % Calculate corners
    \pgfmathsetmacro{\leftx}{\x - \w/2}
    \pgfmathsetmacro{\rightx}{\x + \w/2}
    \pgfmathsetmacro{\bottomy}{\y - \h/2}
    \pgfmathsetmacro{\topy}{\y + \h/2}

    % Draw left edge (bottom-left to top-left)
    \draw[thin, dotted] (\leftx, \bottomy) -- (\leftx, \topy);

    % Draw top edge (top-left to top-right)
    \draw[thin, dotted] (\leftx, \topy) -- (\rightx, \topy);

    % Draw bottom edge (bottom-right to bottom-left)
    \draw[thin, dotted] (\rightx, \bottomy) -- (\leftx, \bottomy);
    \node at (\x, \y) {\scriptsize $I_1$};

\end{tikzpicture}
}
\caption{Schedule representation}
\label{subfig:schedule_efficient}
\end{subfigure}
\caption{Different representations of an efficient solution.}
\label{fig:solutions_efficient}
\end{figure}
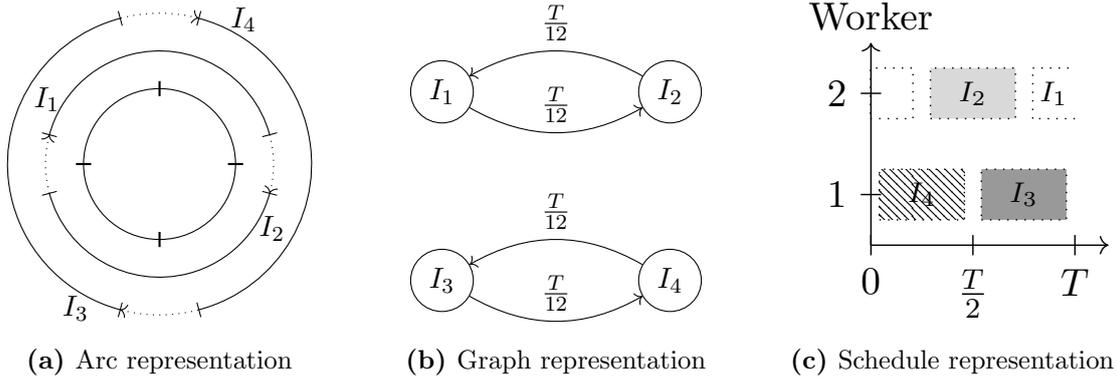

In general periodic assignments, workers corresponding to different cycles are consistently performing disjoint subsets of tasks (see \autoref{fig:solutions_efficient}). To enforce \emph{fairness}—that all workers perform exactly the same sequence of tasks—we require that $\mathcal{M}$ forms a single directed cycle covering all tasks, i.e., a Hamiltonian cycle. This gives rise to the fair variant of the PAP:

\begin{definition}\label{def:fpap}
Given a period $T \in \mathbb{N}$, a set of $T$-periodic tasks $\mathcal{I}$, and a transition graph $\mathcal{G} = (\mathcal{I}, \mathcal{A})$, the \emph{Fair Periodic Assignment Problem (FPAP)} is to find a subset of arcs $\mathcal{M} \subseteq \mathcal{A}$ such that:
\begin{enumerate}[a)]
    \item The total transition time $\sum_{(i,j) \in \mathcal{M}} c_{ij}$ is minimized,
    \item $\mathcal{M}$ defines a Hamiltonian cycle in $\mathcal{G}$.
\end{enumerate}
\end{definition}

Clearly, the FPAP is a special case of the classical Traveling Salesman Problem (TSP), where tasks act as cities and transition times define pairwise distances.

\autoref{fig:solutions_fair} displays three representations of a fair solution to a FPAP instance featuring the same set of tasks as \autoref{fig:solutions_efficient}. Since a fair solution featuring two workers is not possible, the fair solution requires the use of long transition arcs, as shown explicitly in Figures~\ref{subfig:arc_fair} and~\ref{subfig:graph_fair}. The schedule representation in Figure~\ref{subfig:schedule_fair} directly shows that the fair assignment requires three workers. Note that the period of the schedule is longer than $T$, the period of the instance.

\newcommand{\arc}[4]{
    % #1: x-coordinate, #2: y-coordinate, #3: radius
    \pgfmathsetmacro{\innerR}{#3-0.1}
    \pgfmathsetmacro{\outerR}{#3+0.1}
    \pgfmathsetmacro{\middle}{0.1 * #1 + 0.9 * #2}
    \draw[black] ({#3*cos(#1)}, {#3*sin(#1)}) arc (#1:#2:#3);
    \draw[black] ({\innerR *cos(#1)}, {\innerR * sin(#1)}) -- ({\outerR*cos(#1)}, {\outerR*sin(#1)});
    \draw[black] ({\innerR *cos(#2)}, {\innerR*sin(#2)}) -- ({\outerR*cos(#2)}, {\outerR*sin(#2)});
    \node at ({(\outerR + 0.11) *cos(\middle)}, {(\outerR + 0.11) *sin(\middle)}) {#4};
}

\newcommand{\transitionarc}[4]{
    % #1: start angle, #2: end angle, #3: start radius, #4: end radius
    \pgfmathsetmacro{\steps}{30}
    \pgfmathsetmacro{\deltaAngle}{(#2 - #1)/\steps}
    \pgfmathsetmacro{\deltaRadius}{(#4 - #3)/\steps}
    \foreach \i in {0,...,\steps} {
        \pgfmathsetmacro{\angleA}{#1 + \i * \deltaAngle}
        \pgfmathsetmacro{\angleB}{#1 + (\i + 1) * \deltaAngle}
        \pgfmathsetmacro{\radiusA}{#3 + \i * \deltaRadius}
        \pgfmathsetmacro{\radiusB}{#3 + (\i + 1) * \deltaRadius}
        \draw[loosely dotted] 
            ({\radiusA*cos(\angleA)}, {\radiusA*sin(\angleA)}) -- 
            ({\radiusB*cos(\angleB)}, {\radiusB*sin(\angleB)});
    }

   % Arrowhead at the START point, pointing reverse (counterclockwise tangent)
    \pgfmathsetmacro{\startAngle}{#1}
    \pgfmathsetmacro{\startRadius}{#3}

    % Small step along the arc forward (clockwise)
    \pgfmathsetmacro{\smallStepAngle}{\startAngle + 5}  % 5 degrees ahead (adjust if needed)
    \pgfmathsetmacro{\smallStepRadius}{\startRadius + 0.02} % slightly larger radius for smoothness

    % Position where arrow starts (slightly forward on the arc)
    \coordinate (arrowstart) at ({\smallStepRadius*cos(\smallStepAngle)}, {\smallStepRadius*sin(\smallStepAngle)});

    % Vector direction: from arrowstart back to start point (to get reversed arrow)
    \coordinate (arrowtip) at ({\startRadius*cos(\startAngle)}, {\startRadius*sin(\startAngle)});

    \draw[dotted, ->] (arrowstart) -- (arrowtip);
}
   
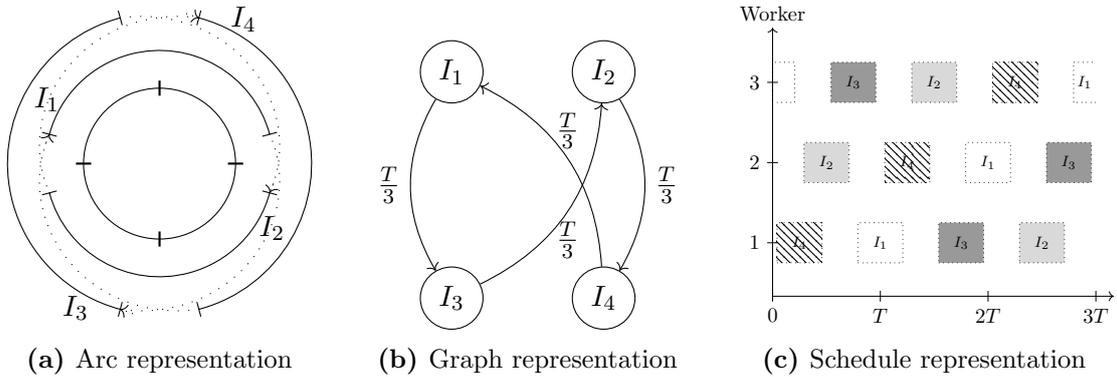
\begin{figure}[htbp!]
\begin{subfigure}[b]{0.3\textwidth}
\centering
\begin{tikzpicture}
     % Draw the full circle.
    \draw[black] (0, 0) circle(1);

    % Draw arcs.
    \arc{15}{165}{1.5}{$I_1$};
    \arc{195}{345}{1.5}{$I_2$};
    \arc{105}{255}{2}{$I_3$};
    \arc{285}{435}{2}{$I_4$};

    % Draw transition arcs.
    \transitionarc{165}{285}{1.5}{2}    % A4 to A1
    \transitionarc{255}{375}{2}{1.5}    % A1 to A3
    \transitionarc{-15}{105}{1.5}{2}    % A3 to A2
    \transitionarc{75}{195}{2}{1.5}    % A2 to A4
    
     % Draw time markers.
    \foreach \angle/\label in {
    90/$0$,
    0/{$\frac{T}{4}$},
    270/{$\frac{T}{2}$},
    180/{$\frac{3T}{4}$}
    } 
    {
    \draw[thick] ({0.9*cos(\angle)}, {0.9*sin(\angle)}) -- ({1.1*cos(\angle)}, {1.1*sin(\angle)});
    %\node at ({1.1*cos(\angle)}, {1.1*sin(\angle)}) {\label};
    }
\end{tikzpicture}
\caption{Arc representation}
\label{subfig:arc_fair}
\end{subfigure}
\hfill
\begin{subfigure}[b]{0.3\textwidth}
\centering
\begin{tikzpicture}
     % Tasks
    \node[draw=black, shape=circle] at (0, 3) (1) {$I_1$};
    \node[draw=black, shape=circle] at (2, 3) (2) {$I_2$};
    \node[draw=black, shape=circle] at (0, 0) (3) {$I_3$};
    \node[draw=black, shape=circle] at (2, 0) (4) {$I_4$};
    
    % Connections.
    \draw[->, bend right] (1) to node[midway, left] {$\frac{T}{3}$} (3);
    \draw[->, bend right] (3) to node[midway, below] {$\frac{T}{3}$} (2);
    \draw[->, bend left] (2) to node[midway, right] {$\frac{T}{3}$} (4);
    \draw[->, bend right] (4) to node[midway, above] {$\frac{T}{3}$} (1);
    
\end{tikzpicture}
\caption{Graph representation}
\label{subfig:graph_fair}
\end{subfigure}
\hfill
\begin{subfigure}[b]{0.35\textwidth}
\centering
\resizebox{\linewidth}{!}{
\begin{tikzpicture}
    
    % Draw payoff axis.
    \draw[->, thin] (0, 0.5)--(6.33, 0.5);
    \draw[->, thin] (0, 0.5)--(0, 5.52) node[above]{Worker};

    % Draw axis labels.
    \draw[thin] (0, 0.4) node[below]{$0$}--(0, 0.6);
    \draw[thin] (2, 0.4) node[below]{$T$}--(2, 0.6);
    \draw[thin] (4, 0.4) node[below]{$2T$}--(4, 0.6);
    \draw[thin] (6, 0.4) node[below]{$3T$}--(6, 0.6);

    \draw[thin] (-0.1, 1.5) node[left]{$1$}--(0.1, 1.5);
    \draw[thin] (-0.1, 3) node[left]{$2$}--(0.1, 3);
    \draw[thin] (-0.1, 4.5) node[left]{$3$}--(0.1, 4.5);

     % Draw tasks worker 1.
     \node[draw, thin, dotted, minimum width=0.8333cm, minimum height=0.75cm, align=center, pattern=north west lines] at (0.5, 1.5) {\scriptsize{$I_4$}};
     \node[draw, thin, dotted, minimum width=0.8333cm, minimum height=0.75cm, align=center] at (2, 1.5) {\scriptsize{$I_1$}};
     \node[draw, thin, dotted, minimum width=0.8333cm, minimum height=0.75cm, align=center, fill=black!40] at (3.5, 1.5) {\scriptsize{$I_3$}};
     \node[draw, thin, dotted, minimum width=0.8333cm, minimum height=0.75cm, align=center, fill=black!15] at (5, 1.5) {\scriptsize{$I_2$}};

     % Draw tasks worker 2.
     \node[draw, thin, dotted, minimum width=0.8333cm, minimum height=0.75cm, align=center, fill=black!15] at (1, 3) {\scriptsize{$I_2$}};
     \node[draw, thin, dotted, minimum width=0.8333cm, minimum height=0.75cm, align=center, pattern=north west lines] at (2.5, 3) {\scriptsize{$I_4$}};
     \node[draw, thin, dotted, minimum width=0.8333cm, minimum height=0.75cm, align=center] at (4, 3) {\scriptsize{$I_1$}};
     \node[draw, thin, dotted, minimum width=0.8333cm, minimum height=0.75cm, align=center, fill=black!40] at (5.5, 3) {\scriptsize{$I_3$}};

     % Draw tasks worker 1.
     \node[draw, thin, dotted, minimum width=0.4166cm, minimum height=0.75cm, align=center] at (0.2033, 4.5) {};
     \node[draw, thin, dotted, minimum width=0.8333cm, minimum height=0.75cm, align=center, fill=black!40] at (1.5, 4.5) {\scriptsize{$I_3$}};
     \node[draw, thin, dotted, minimum width=0.8333cm, minimum height=0.75cm, align=center, fill=black!15] at (3, 4.5) {\scriptsize{$I_2$}};
     \node[draw, thin, dotted, minimum width=0.8333cm, minimum height=0.75cm, align=center, pattern=north west lines] at (4.5, 4.5) {\scriptsize{$I_4$}};
     %\node[draw, thin, dotted, minimum width=0.8333cm, minimum height=0.75cm, align=center] at (6, 4.5) {\scriptsize{$I_1$}};

    \def\x{5.7933}
    \def\y{4.5}
    \def\w{0.41333}
    \def\h{0.75}

    % Calculate corners
    \pgfmathsetmacro{\leftx}{\x - \w/2}
    \pgfmathsetmacro{\rightx}{\x + \w/2}
    \pgfmathsetmacro{\bottomy}{\y - \h/2}
    \pgfmathsetmacro{\topy}{\y + \h/2}

    % Draw left edge (bottom-left to top-left)
    \draw[thin, dotted] (\leftx, \bottomy) -- (\leftx, \topy);

    % Draw top edge (top-left to top-right)
    \draw[thin, dotted] (\leftx, \topy) -- (\rightx, \topy);

    % Draw bottom edge (bottom-right to bottom-left)
    \draw[thin, dotted] (\rightx, \bottomy) -- (\leftx, \bottomy);
    \node at (\x, \y) {\scriptsize $I_1$};
\end{tikzpicture}
}
\caption{Schedule representation}
\label{subfig:schedule_fair}
\end{subfigure}
\caption{Different representations of a fair solution.}
\label{fig:solutions_fair}
\end{figure}

Finally, we introduce a lower bound on the number of workers required in any periodic assignment. Let $\tasks(t)$ denote the number of tasks that intersect time $t \in [0, T)$, and define the system’s \emph{load} as
$$L:=\max_{t\in[0,T)} \tasks(t).$$
The instance of \autoref{fig:solutions_efficient} has a load of $L=2$, matching the number of workers in the efficient solution.

\subsection{Background and Related Literature}
In the special case where there is some $t'$ such that $\tasks(t')=0$, \cite{korst1994periodic} show that PAP reduces to coloring an interval graph, yielding a periodic assignment with $L$ workers, matching the lower bound. \cite{gupta} provide an $\mathcal{O}(n \log n)$ algorithm for computing such a coloring. Also the FPAP is easily solved in this case: at time $t'$ all workers are idling, so $L$ workers can cycle through the $L$ colors to cover all tasks, again matching the lower bound. 

In the general case, $L$ workers still suffice for PAP, as follows from Dilworth’s theorem \cite{orlin1982minimizing}. An optimal assignment can be found via linear programming, the Hungarian algorithm, or the $\mathcal{O}(n^2 \log n)$ algorithm of \cite{vanlieshout2021integrated}.

To the best of our knowledge, we are the first to study the FPAP. However, a similar problem is considered in \cite{gachet_et_al:OASIcs.ATMOS.2024.5}, which was presented at ATMOS 2024 and served as a direct inspiration for this work. Rather than requiring strict fairness, the authors consider \textit{balanced} assignments that are asymptotically fair, i.e., where workers perform the same work in the long-run average. They reduce the problem to a construction involving pebbles on arc-colored Eulerian graphs, showing that if a balanced assignment for $w$ workers exists, it can be found in polynomial time. Their initial construction yields a period bounded by $w^2 \cdot w!$, but in a subsequent extension \cite{gachet2025balancedassignmentsperiodictasks}, the authors improve this to a linear period. 

Other special cases of the TSP are surveyed in \cite{burkard1998well} and other cyclic scheduling problems are discussed in \cite{LEVNER2010352}.

\subsection{Main Contributions}
The main contributions of this paper are fourfold. First, we present an $\mathcal{O}(n \log n)$ exact algorithm for the PAP, improving upon the previous best $\mathcal{O}(n^2 \log n)$ runtime. Second, we prove that the classical Nearest Neighbor heuristic for TSP yields a fair periodic assignment requiring at most $L + 1$ workers, implying that the \emph{price of fairness} is at most one additional worker. This bound is tight. Third, we develop an $\mathcal{O}(n \log n)$ exact algorithm for the FPAP based on subtour patching. Fourth, we show that allowing longer assignment periods and requiring balancedness rather than fairness does not reduce the number of required workers.

\section{Periodic Assignment}
\label{sec:periodic}

In this section, we characterize optimal solutions to PAP, which we use later on to solve the FPAP. Moreover, we present a new $\mathcal{O}(n\log n)$ algorithm for solving PAP.

\subsection{Theory}
\label{subsec:periodic_theory}
We begin by defining notation relevant to the PAP. Transition arc $(i,j)\in \mathcal{A}$ corresponds to the $T$-periodic interval $[b_i,a_j]$. Analogous to $\tasks(t)$, for a periodic assignment $\mathcal{M}\subseteq \mathcal{A}$, let $\mathcal{M}(t)$ denote the number of transition arcs in $\mathcal{M}$ whose associated interval intersects time $t\in [0,T)$, and let  

$$c(\tasks) := \sum_{i\in \tasks} c(i) = \int_0^T \tasks(t)dt \quad \text{ and } \quad c(\mathcal{M}) := \sum_{(i,j)\in \mathcal{M}}c_{ij} = \int_0^T \mathcal{M}(t)dt. $$
Since workers are either performing tasks or transitioning between tasks, it holds that 
the sum $\tasks(t)+\mathcal{M}(t)$ equals the number of workers required to operate the schedule for all $t\in [0,T)$, which also equals $\left(c(\tasks)+c(\mathcal{M})\right)/T$. Combining insights from \cite{orlin1982minimizing} and \cite{vanlieshout2021integrated}, we now characterize the optimal periodic assignment, which uses exactly $L$ workers. 

\begin{theorem}
    \label{thm:pap}
    Let $\mathcal{M}^*\subseteq \mathcal{A}$ denote a feasible periodic assignment. The following statements are equivalent:
    \begin{enumerate}
        \item[(i)] $\mathcal{M}^*$ is an optimal solution to PAP,
        \item[(ii)] $c(\tasks)+c(\mathcal{M^*})=LT$,
        \item[(iii)] $\tasks(t)+\mathcal{M^*}(t) = L$ for all $t \in [0,T)$,
        \item[(iv)] $\mathcal{M^*}(t') = 0$ for some $t'\in [0,T)$.
    \end{enumerate}
\end{theorem}
\begin{proof}
    $(iv) \Rightarrow (iii)$: Since all workers are either busy performing tasks or transitioning, the total number of workers at any time $t$ is $\tasks(t) + \mathcal{M}^*(t)$. By assumption, there exists a $t' \in [0, T)$ such that $\mathcal{M}^*(t') = 0$, implying $\tasks(t') \leq L$. Conversely, since $\tasks(t)$ attains its maximum at some $t''$, we have $\tasks(t'') = L$, so $\tasks(t'') + \mathcal{M}^*(t'') \geq L$. Because the total number of workers is constant over time, we must have that $
    \tasks(t) + \mathcal{M}^*(t) = L \text{ for all } t \in [0,T).$

    $(iii) \Rightarrow (ii)$: Integrating over the interval $[0,T)$ yields:
    \[
    c(\tasks) + c(\mathcal{M}^*) = \int_0^T \tasks(t) + \mathcal{M}^*(t) \, dt = \int_0^T L \, dt = LT.
    \]

    $(ii) \Rightarrow (i)$: If $c(\tasks) + c(\mathcal{M}^*) = LT$, then the schedule uses exactly $L$ workers over time, matching the theoretical lower bound. Hence, $\mathcal{M}^*$ must be optimal.

    $(i) \Rightarrow (iv)$: Suppose, for contradiction, that $\mathcal{M}^*(t) \geq 1$ for all $t \in [0, T)$. Then, there exists a subset of transition arcs whose intervals fully cover $[0, T)$. Without loss of generality, let these intervals be $\mathcal{N}^* = \{[b_1, a_1], \ldots, [b_m, a_m]\}$, where 
        \begin{equation*}
        {b_1} < {a_m} < {b_2} < {a_1} < \ldots < {b_m} < {a_{m-1}}.
    \end{equation*}
    We can now construct a shorter matching $\mathcal{N}' = \{[b_1, a_m], [b_2, a_1], \ldots, [b_m, a_{m-1}]\}$, which reduces the total transition time. This contradicts the optimality of $\mathcal{M}^*$. Therefore, $\mathcal{M}^*(t) = 0$ must hold for some $t \in [0, T)$.
\end{proof}

\subsection{Algorithms}
\label{subsec:periodic_algo}
Every transition arc $(i,j)\in \mathcal{A}$ in the PAP corresponds to a periodic interval $[b_i, a_j]$, and its cost depends only on the length of this interval. Therefore, solving PAP is equivalent to matching each end time $b_i$ to a start time $a_j$ such that the resulting intervals minimize the total transition time. 

Crucially, the identity of individual tasks does not affect the optimality of the assignment — only the multiset of start and end times matters. This insight enables an efficient algorithm, which we call \textsc{Shift-Sort-and-Match}. It improves upon the \(\mathcal{O}(n^2 \log n)\) method of \cite{vanlieshout2021integrated} by reducing the complexity to \(\mathcal{O}(n \log n)\), primarily by exploiting this structural invariance.

Algorithm~\ref{alg:sorted_greedy} outlines the procedure. The algorithm first shifts all start and end times such that \(\tasks(0)=L\), the maximum number of simultaneously active tasks. It then sorts all start and end times into a single array, and greedily matches each end time to the earliest available start time. Note that the algorithm computes an optimal assignment in terms of a matching between end and start times (still denoted by \(\mathcal{M}\)); the task-to-task assignment can be recovered by tracing these matched time points back to their associated tasks.

The shifting step of Algorithm~\ref{alg:sorted_greedy} requires both the maximum $L$ and a maximizer $t^*$ of $\tasks(t)$. After sorting all events, these can be computed by sweeping through the timeline and maintaining a counter that increments at every start time and decrements at every end time, while keeping track of the interval $(a_i,b_j)$ where the where the counter reaches its highest value. The maximum load $L$ is then given by the largest value observed by the counter. Finally, any point within the interval $(a_i, b_j)$ where this maximum occurs can be chosen as the maximizer $t^*$. This procedure takes $\mathcal{O}(n\log n)$ due to the initial sorting step. 

\begin{algorithm}
\caption{\textsc{Shift-Sort-and-Match}}
\label{alg:sorted_greedy}
\begin{algorithmic}[1]
\STATE Shift all start and end times such that $\tasks(0) = L$
\STATE Construct array $R$ with all starts and ends, each tagged with its type
\STATE Sort $R$ in increasing order; break ties by placing ends before starts
\STATE Initialize empty stack $Q$ for unmatched ends
\STATE Initialize matching $\mathcal{M} \gets \emptyset$
\STATE Set $i \gets 1$
\WHILE{$i \leq |R|$}
    \IF{$R[i]$ is a start}
        \STATE Pop top element $e$ from stack $Q$
        \STATE Add arc $(e, R[i])$ to $\mathcal{M}$
        \STATE $i \gets i + 1$
    \ELSIF{$R[i+1]$ is a start}
        \STATE Add arc $(R[i], R[i+1])$ to $\mathcal{M}$
        \STATE $i \gets i + 2$
    \ELSE
        \STATE Push $R[i]$ onto stack $Q$
        \STATE $i \gets i + 1$
    \ENDIF
\ENDWHILE
\STATE \textbf{return} $\mathcal{M}$
\end{algorithmic}
\end{algorithm}

\begin{theorem}
The \textsc{Shift-Sort-and-Match} algorithm computes an optimal periodic assignment \(\mathcal{M}^*\) for the Periodic Assignment Problem in \(\mathcal{O}(n \log n)\) time.
\end{theorem}

\begin{proof}
We begin by analyzing the time complexity. Recall that computing $L$ and the maximizer $t^*$ takes $\mathcal{O}(n \log n)$. Sorting $2n$ time points also takes $\mathcal{O}(n \log n)$, and the main loop of the algorithm executes $\mathcal{O}(n)$ operations, each taking $\mathcal{O}(1)$ time. Hence, the total runtime is $\mathcal{O}(n \log n)$.

To prove correctness, we show that the algorithm constructs a feasible assignment satisfying condition $(iv)$ from Theorem~\ref{thm:pap}, which implies optimality. 

We first show that the stack size when processing an event at time $t$ equals $L-\tasks(t)$. Initially, the stack is empty and $\tasks(t)=L$. Each step of the algorithm maintains this invariant through the following cases:

\begin{itemize}
    \item \textbf{Case 1 (Line 8):} If \(R[i]\) is a start, then \(\tasks(t)\) increases by 1, and the stack size decreases by 1.
    \item \textbf{Case 2 (Line 12):} If \(R[i+1]\) is a start, then \(\tasks(t)\) first decreases by 1 (due to the end \(R[i]\)), then increases by 1 (due to the start \(R[i+1]\)); the stack remains unchanged.
    \item \textbf{Case 3 (Line 15):} Otherwise, \(R[i]\) is an end; \(\tasks(t)\) decreases by 1, and the stack size increases by 1.
\end{itemize}

This invariant yields the following consequences:
\begin{enumerate}
    \item The stack size remains nonnegative at all times.
    \item When processing a start, the stack must be nonempty, ensuring a valid match.
    \item At termination, \(\tasks(t) = L\) again, so the stack is empty.
\end{enumerate}

Together, these observations confirm that all starts are matched to ends—either directly via Line~13 or by popping from the stack in Line~9—ensuring feasibility of the assignment.

For optimality, note that all arcs match from a lower index to a higher index. In particular, no transition arcs `loop' around time $t=0$. Hence $\mathcal{M}^*(0)=0$, which implies the solution is optimal by property $(iv)$ of \autoref{thm:pap}.
\end{proof}

\section{Fair Periodic Assignment}
\label{sec:fair}

Imposing fairness in the periodic assignment problem amounts to requiring that the assignment $(\mathcal{M})$ defines a Hamiltonian cycle in the transition graph. Unlike in the standard PAP, this condition makes the specific identities of transition arcs essential, as these determine whether there are any subtours or not. 

This additional constraint makes the problem strictly harder: it is no longer guaranteed that a solution exists using only \(L\) workers. As illustrated in \autoref{fig:solutions_fair}, certain FPAP instances necessarily involve long transition arcs, forcing any fair solution to use \(L+1 = 3\) workers instead of \(L = 2\). Remarkably, we show that this is the worst-case scenario: every FPAP instance can be solved using either \(L\) or \(L+1\) workers.

This section proceeds as follows. First, we prove that the Nearest Neighbor heuristic always yields a fair solution with at most \(L+1\) workers. Then, we characterize the class of FPAP instances that admit a fair solution using exactly \(L\) workers. Building on this result, we then present an exact algorithm for solving FPAP.

\subsection{Nearest Neighbor}
The \textsc{Nearest Neighbor} heuristic is a classical heuristic for the TSP, and repeatedly matches the current task with the closest unvisited task until all tasks are visited. This is formalized in Algorithm~\ref{alg:nearest_Neighbor}.

\begin{algorithm}
\caption{\textsc{Nearest Neighbor}}
\label{alg:nearest_Neighbor}
\begin{algorithmic}[1]
\STATE Set current task $i$ to arbitrary task $u$
\STATE Store unvisited tasks $\mathcal{I} \setminus \{i\}$ in $\mathcal{U}$
\STATE Sort $\mathcal{U}$ based on increasing start time
\STATE Initialize matching $\mathcal{M} \leftarrow \emptyset$
\WHILE{$|\mathcal{U}| \geq 1$}
    \STATE Let $v$ be task in $\mathcal{U}$ closest to $i$
    \STATE Add arc $(i, v)$ to $\mathcal{M}$
    \STATE Remove $v$ from $\mathcal{U}$
    \STATE $i \leftarrow v$
\ENDWHILE 
\STATE Add arc $(i, u)$ to $\mathcal{M}$
\STATE \textbf{return} $\mathcal{M}$
\end{algorithmic}
\end{algorithm}

To demonstrate the heuristic, we consider the same instance as in \autoref{fig:solutions_efficient} but now without task $I_2$. \autoref{fig:nn} illustrates that \textsc{Nearest Neighor} returns a solution with three workers, while the optimal solution only uses two. As the following theorem shows, this corresponds to the worst-case performance of this heuristic.

\newcommand{\arc}[4]{
    % #1: x-coordinate, #2: y-coordinate, #3: radius
    \pgfmathsetmacro{\innerR}{#3-0.1}
    \pgfmathsetmacro{\outerR}{#3+0.1}
    \pgfmathsetmacro{\middle}{0.1 * #1 + 0.9 * #2}
    \draw[black] ({#3*cos(#1)}, {#3*sin(#1)}) arc (#1:#2:#3);
    \draw[black] ({\innerR *cos(#1)}, {\innerR * sin(#1)}) -- ({\outerR*cos(#1)}, {\outerR*sin(#1)});
    \draw[black] ({\innerR *cos(#2)}, {\innerR*sin(#2)}) -- ({\outerR*cos(#2)}, {\outerR*sin(#2)});
    \node at ({(\outerR + 0.11) *cos(\middle)}, {(\outerR + 0.11) *sin(\middle)}) {#4};
}

\newcommand{\transitionarc}[4]{
    % #1: start angle, #2: end angle, #3: start radius, #4: end radius
    \pgfmathsetmacro{\steps}{30}
    \pgfmathsetmacro{\deltaAngle}{(#2 - #1)/\steps}
    \pgfmathsetmacro{\deltaRadius}{(#4 - #3)/\steps}
    \foreach \i in {0,...,\steps} {
        \pgfmathsetmacro{\angleA}{#1 + \i * \deltaAngle}
        \pgfmathsetmacro{\angleB}{#1 + (\i + 1) * \deltaAngle}
        \pgfmathsetmacro{\radiusA}{#3 + \i * \deltaRadius}
        \pgfmathsetmacro{\radiusB}{#3 + (\i + 1) * \deltaRadius}
        \draw[loosely dotted] 
            ({\radiusA*cos(\angleA)}, {\radiusA*sin(\angleA)}) -- 
            ({\radiusB*cos(\angleB)}, {\radiusB*sin(\angleB)});
    }

   % Arrowhead at the START point, pointing reverse (counterclockwise tangent)
    \pgfmathsetmacro{\startAngle}{#1}
    \pgfmathsetmacro{\startRadius}{#3}

    % Small step along the arc forward (clockwise)
    \pgfmathsetmacro{\smallStepAngle}{\startAngle + 5}  % 5 degrees ahead (adjust if needed)
    \pgfmathsetmacro{\smallStepRadius}{\startRadius + 0.02} % slightly larger radius for smoothness

    % Position where arrow starts (slightly forward on the arc)
    \coordinate (arrowstart) at ({\smallStepRadius*cos(\smallStepAngle)}, {\smallStepRadius*sin(\smallStepAngle)});

    % Vector direction: from arrowstart back to start point (to get reversed arrow)
    \coordinate (arrowtip) at ({\startRadius*cos(\startAngle)}, {\startRadius*sin(\startAngle)});

    \draw[dotted, ->] (arrowstart) -- (arrowtip);
}
   
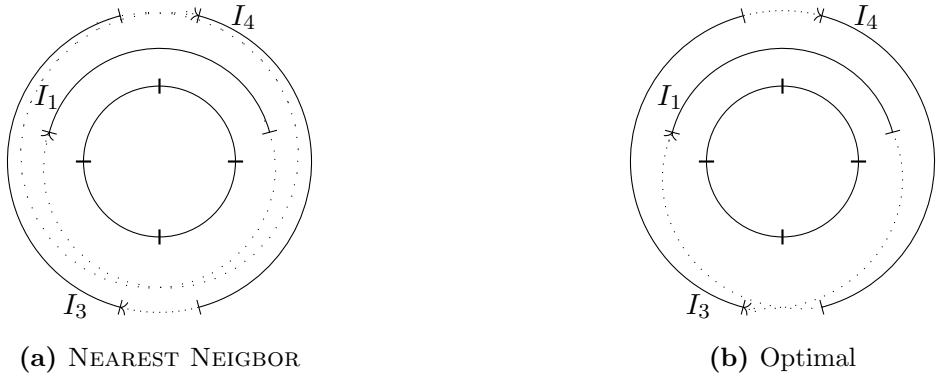
\begin{figure}[htbp!]
\begin{subfigure}[b]{0.45\textwidth}
\centering
\begin{tikzpicture}
     % Draw the full circle.
    \draw[black] (0, 0) circle(1);

     % Draw arcs.
    \arc{15}{165}{1.5}{$I_1$};
    \arc{105}{255}{2}{$I_3$};
    \arc{285}{435}{2}{$I_4$};

    % Draw transition arcs.
    \draw[dotted, <-] ({2*cos(255)}, {2*sin(255)}) arc (255:285:2);

    % Draw transition arcs.
    \transitionarc{75}{375}{2}{1.5}   
    \transitionarc{165}{465}{1.5}{2}
    
     % Draw time markers.
    \foreach \angle/\label in {
    90/$0$,
    0/{$\frac{T}{4}$},
    270/{$\frac{T}{2}$},
    180/{$\frac{3T}{4}$}
    } 
    {
    \draw[thick] ({0.9*cos(\angle)}, {0.9*sin(\angle)}) -- ({1.1*cos(\angle)}, {1.1*sin(\angle)});
    }
\end{tikzpicture}
\caption{\textsc{Nearest Neigbor}}
\end{subfigure}
\hfill
\begin{subfigure}[b]{0.45\textwidth}
\centering
\begin{tikzpicture}
    % Draw the full circle.
    \draw[black] (0, 0) circle(1);

    % Draw arcs.
    \arc{15}{165}{1.5}{$I_1$};
    \arc{105}{255}{2}{$I_3$};
    \arc{285}{435}{2}{$I_4$};

    % Draw transition arcs.
    \draw[dotted, <-] ({2*cos(435)}, {2*sin(435)}) arc (75:105:2);
    \transitionarc{165}{285}{1.5}{2}    % A4 to A1
    \transitionarc{255}{375}{2}{1.5}    % A1 to A3

    % Draw time markers.
    \foreach \angle/\label in {
    90/$0$,
    0/{$\frac{T}{4}$},
    270/{$\frac{T}{2}$},
    180/{$\frac{3T}{4}$}
    } 
    {
    \draw[thick] ({0.9*cos(\angle)}, {0.9*sin(\angle)}) -- ({1.1*cos(\angle)}, {1.1*sin(\angle)});
    }
\end{tikzpicture}
\caption{Optimal}
\end{subfigure}
\caption{\textsc{Nearest Neighbor} solution versus optimal solution.}
\label{fig:nn}
\end{figure}

\begin{theorem}
The \textsc{Nearest Neighbor} algorithm returns a fair periodic assignment requiring at most $L+1$ workers in $\mathcal{O}(n \log n)$ time.
\end{theorem}

\begin{proof}
We first analyze the runtime. A crucial observation is that the distance of the current task to the next depends solely on the end time of the current task and start time of the next task. If unvisited tasks are stored in increasing order of start time, one can efficiently find the closest unvisited task. To this end, we store the univisited tasks $\mathcal{U}$ in a self-balancing binary search tree, allowing us to maintain a fixed task order even as tasks are visited and removed throughout the course of the algorithm. Constructing this tree takes $\mathcal{O}(n \log n)$. Finding the next task $v$ and removing it from $\mathcal{U}$ both take $\mathcal{O}(\log n)$. Since $\mathcal{O}(n)$ operations are required until all tasks are visited, the total time complexity is $\mathcal{O}(n \log n)$.

\textsc{Nearest Neighbor} always returns a fair periodic assignment. It remains to show that this assignment requires at most $L + 1$ workers. Suppose, to the contrary, that the assignment requires more than $L + 1$ workers. Let $t=0$ correspond to the start of the first task that is visited. If the assignment requires more than $L+1$ workers, there is some task $i=(a,b)$ that is started after $L$ periods, i.e., after $L$ complete revolutions around the circle. This means that in the first $L$ periods, a task is performed at time $a$. Together with task $i$, this implies the existence of some small $\varepsilon>0$ such that $\mathcal{I}(a+\varepsilon) = L + 1$, a contradiction.
\end{proof}

The performance guarantee of \textsc{Nearest Neighbor} allows us to upper-bound the number of additional workers required to operate a fair assignment, i.e., the price of fairness.
\begin{corollary}
\label{col:pof}
The price of fairness is $1/L$.
\end{corollary}
\begin{proof}
An efficient assignment always requires exactly $L$ workers. \textsc{Nearest Neighbor} returns a fair solution with at most $L + 1$ workers. This upper bound is tight, see, e.g., \autoref{fig:solutions_fair}. It follows that the price of fairness equals $\frac{(L + 1) - L}{L} = \frac{1}{L}$.
\end{proof}

\subsection{Theory}
\label{subsec:fair_theory}

We now aim to characterize the instances that admit a solution that is both fair and \textit{efficient}, i.e., require no more than $L$ workers. Remark that, in any efficient solution, workers are only idle during periods in which the load is not maximal, i.e., for which $\mathcal{I}(t) < L$. We refer to such intervals as \textit{idle intervals}. Since a worker is idle right before and after performing a task, tasks act as connections between idle intervals. Moreover, the transition arcs in any efficient solution are fully contained in idle intervals. As transition arcs determine which tasks are performed consecutively, they implicitly define which idle intervals are visited along each (sub)tour. 

In case a solution consists of multiple subtours, we distinguish two cases. If two disjoint cycles share transition arcs in the same idle interval, exchanging the end tasks between two transition arcs in this interval may \textit{patch} the cycles together, reducing the number of subtours by one without increasing the transition costs and bringing us closer to a fair solution. If they are not simultaneously idle, however, such a procedure is impossible and a fair solution with $L$ workers is out of reach. In this case, the disjoint cycles can only be patched together by transition arcs that are active outside of idle intervals, implying the use of at least one additional worker compared to an efficient solution. The connectivity of idle intervals thus contains crucial information regarding the existence of a fair solution with $L$ workers. 

In the remainder of this section, we show that the way in which tasks connect idle intervals provides a necessary and sufficient condition for the existence of solutions that are both efficient and fair. In particular, we show that such solutions can always be constructed from efficient solutions through patching.

\newcommand{\arc}[4]{
    % #1: x-coordinate, #2: y-coordinate, #3: radius
    \pgfmathsetmacro{\innerR}{#3-0.1}
    \pgfmathsetmacro{\outerR}{#3+0.1}
    \pgfmathsetmacro{\middle}{0.1 * #1 + 0.9 * #2}
    \draw[black] ({#3*cos(#1)}, {#3*sin(#1)}) arc (#1:#2:#3);
    \draw[black] ({\innerR *cos(#1)}, {\innerR * sin(#1)}) -- ({\outerR*cos(#1)}, {\outerR*sin(#1)});
    \draw[black] ({\innerR *cos(#2)}, {\innerR*sin(#2)}) -- ({\outerR*cos(#2)}, {\outerR*sin(#2)});
    \node at ({(\outerR + 0.11) *cos(\middle)}, {(\outerR + 0.11) *sin(\middle)}) {#4};
}

\tikzset{->-/.style={decoration={
			markings,
			mark=at position #1 with {\arrow{>}}},postaction={decorate}}}

\newcommand{\slack}[3]{
    % #1: x-start, #2: x-end #3: y
    \draw[thick] (#1, #3) -- (#2, #3);
    \draw[fill=black] (#1, #3) circle (2pt);
    \draw[fill=black] (#2, #3) circle (2pt);
}

\newcommand{\zeroslack}[2]{
    % #1: x-start, #2: x-end 
    \draw[thick] (#1, 0) -- (#2, 0);
    \draw[fill=white] (#1, 0) circle (2pt);
    \draw[fill=white] (#2, 0) circle (2pt);
}

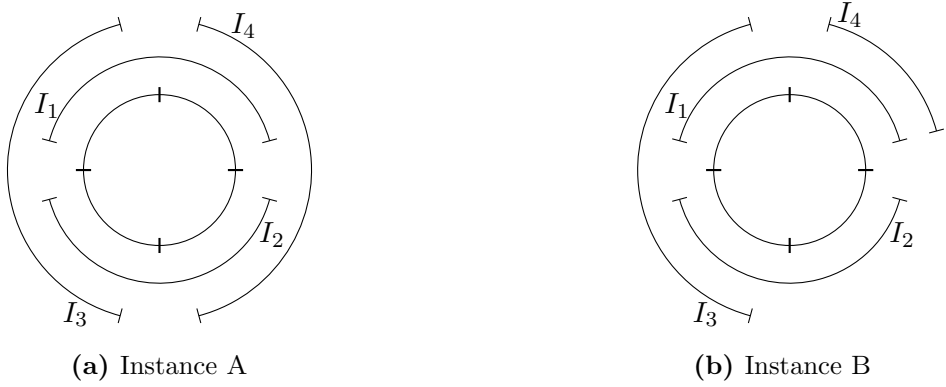
\begin{figure}[htbp!]
\begin{subfigure}[b]{0.45\textwidth}
\centering
\begin{tikzpicture}
    % Draw the full circle.
    \draw[black] (0, 0) circle(1);

    % Draw arcs.
    \arc{15}{165}{1.5}{$I_1$};
    \arc{195}{345}{1.5}{$I_2$};
    \arc{105}{255}{2}{$I_3$};
    \arc{285}{435}{2}{$I_4$};

    % Draw time markers.
    \foreach \angle/\label in {
    90/$0$,
    0/{$\frac{T}{4}$},
    270/{$\frac{T}{2}$},
    180/{$\frac{3T}{4}$}
    } 
    {
    \draw[thick] ({0.9*cos(\angle)}, {0.9*sin(\angle)}) -- ({1.1*cos(\angle)}, {1.1*sin(\angle)});
    }
\end{tikzpicture}
\caption{Instance \textsc{A}}
\label{subfig:instance_unfair}
\end{subfigure}
\hfill
\begin{subfigure}[b]{0.45\textwidth}
\centering
\begin{tikzpicture}
    % Draw the full circle.
    \draw[black] (0, 0) circle(1);

    % Draw arcs.
    \arc{15}{165}{1.5}{$I_1$};
    \arc{195}{345}{1.5}{$I_2$};
    \arc{105}{255}{2}{$I_3$};
    \arc{375}{435}{2}{$I_4$};

    % Draw time markers.
    \foreach \angle/\label in {
    90/$0$,
    0/{$\frac{T}{4}$},
    270/{$\frac{T}{2}$},
    180/{$\frac{3T}{4}$}
    } 
    {
    \draw[thick] ({0.9*cos(\angle)}, {0.9*sin(\angle)}) -- ({1.1*cos(\angle)}, {1.1*sin(\angle)});
    }
\end{tikzpicture}
\caption{Instance \textsc{B}}
\label{subfig:instance_fair}
\end{subfigure}
\caption{Two instances of the periodic assignment problem. Instance \textsc{A} does not admit a fair and efficient solution, while instance \textsc{B} does.}
\label{fig:fair_instances}
\end{figure}

Throughout, we illustrate our analysis on the instances in \autoref{fig:fair_instances}. The instance in \autoref{subfig:instance_unfair} equals that of \autoref{fig:solutions_fair} and does not admit a fair solution with $L=2$ workers. In contrast, the slightly modified instance in \autoref{subfig:instance_fair} does admit a fair and efficient solution. 

We proceed by formalizing the notion of idle interval.
\begin{definition} \label{def:idle_interval}
An idle interval $k$ is a maximal $T$-periodic closed interval $[s_k, e_k]$ satisfying $L - \mathcal{I}(t) > 0$ for all $t \in [s_k, e_k]$.
\end{definition}
An interval is maximal when it is not contained in another interval. It is easy to see that each instance admits at most $n$ idle intervals. Moreover, the start and end time of each interval correspond to the end and start time of some tasks, respectively. The idle intervals of instances \textsc{A} and \textsc{B} are shown in Figures~\ref{subfig:idle_unfair} and~\ref{subfig:idle_fair}, respectively. Observe that the number of idle workers $L - \mathcal{I}(t)$ need not be constant within an idle interval.

\newcommand{\arc}[4]{
    % #1: x-coordinate, #2: y-coordinate, #3: radius
    \pgfmathsetmacro{\innerR}{#3-0.1}
    \pgfmathsetmacro{\outerR}{#3+0.1}
    \pgfmathsetmacro{\middle}{0.1 * #1 + 0.9 * #2}
    \draw[black] ({#3*cos(#1)}, {#3*sin(#1)}) arc (#1:#2:#3);
    \draw[black] ({\innerR *cos(#1)}, {\innerR * sin(#1)}) -- ({\outerR*cos(#1)}, {\outerR*sin(#1)});
    \draw[black] ({\innerR *cos(#2)}, {\innerR*sin(#2)}) -- ({\outerR*cos(#2)}, {\outerR*sin(#2)});
    \node at ({(\outerR + 0.11) *cos(\middle)}, {(\outerR + 0.11) *sin(\middle)}) {#4};
}

\tikzset{->-/.style={decoration={
			markings,
			mark=at position #1 with {\arrow{>}}},postaction={decorate}}}

\newcommand{\slack}[3]{
    % #1: x-start, #2: x-end #3: y
    \draw[thick] (#1, #3) -- (#2, #3);
    \draw[fill=black] (#1, #3) circle (2pt);
    \draw[fill=black] (#2, #3) circle (2pt);
}

\newcommand{\zeroslack}[2]{
    % #1: x-start, #2: x-end 
    \draw[thick] (#1, 0) -- (#2, 0);
    \draw[fill=white] (#1, 0) circle (2pt);
    \draw[fill=white] (#2, 0) circle (2pt);
}

\begin{figure}[htbp!]
\begin{subfigure}[b]{0.45\textwidth}
\centering
\resizebox{\linewidth}{!}{
\begin{tikzpicture}
    
    % Draw payoff axis.
    \draw[->, thin] (0,0)--(8.33,0);
    \draw[->, thin] (0,0)--(0, 2.5) node[above]{$L - \mathcal{I}(t)$};

    % Draw axis labels.
    \draw[thin] (0, -0.1) node[below]{$0$}--(0, 0.1);
    \draw[thin] (2, -0.1) node[below]{$\frac{T}{4}$}--(2, 0.1);
    \draw[thin] (4, -0.1) node[below]{$\frac{T}{2}$}--(4, 0.1);
    \draw[thin] (6, -0.1) node[below]{$\frac{3T}{4}$}--(6, 0.1);
    \draw[thin] (8, -0.1) node[below]{$T$}--(8, 0.1);

    \draw[thin] (-0.1, 0) node[left]{$0$}--(0.1, 0);
    \draw[thin] (-0.1, 1) node[left]{$1$}--(0.1, 1);
    \draw[thin] (-0.1, 2) node[left]{$2$}--(0.1, 2);

    % Draw slack function.
    \draw[dotted, thin] (0.33, 0) -- (0.33, 1);
    \draw[dotted, thin] (1.67, 0) -- (1.67, 1);
    \draw[dotted, thin] (2.33, 0) -- (2.33, 1);
    \draw[dotted, thin] (3.67, 0) -- (3.67, 1);
    \draw[dotted, thin] (4.33, 0) -- (4.33, 1);
    \draw[dotted, thin] (5.67, 0) -- (5.67, 1);
    \draw[dotted, thin] (6.33, 0) -- (6.33, 1);
    \draw[dotted, thin] (7.67, 0) -- (7.67, 1);
    
    \draw[thick] (0, 1) -- (0.33, 1);
    \draw[fill=black] (0.33, 1) circle (2pt);
    \zeroslack{0.33}{1.67};
    \slack{1.67}{2.33}{1};
    \zeroslack{2.33}{3.67};
    \slack{3.67}{4.33}{1};
    \zeroslack{4.33}{5.67};
    \slack{5.67}{6.33}{1};
    \zeroslack{6.33}{7.67};
    \draw[thick] (7.67, 1) -- (8.33, 1);
    \draw[fill=black] (7.67, 1) circle (2pt);

    % Draw clusters.
    \draw[<->] (1.67, 1.5) -- node[midway, above] {$S_1$} (2.33, 1.5);
    \draw[<->] (3.67, 1.5) -- node[midway, above] {$S_2$} (4.33, 1.5);
    \draw[<->] (5.67, 1.5) -- node[midway, above] {$S_3$} (6.33, 1.5);
    \draw[<-] (7.67, 1.5) -- node[midway, above] {$S_4$} (8.33, 1.5);
    \draw[->] (0, 1.5) -- (0.33, 1.5);
    
\end{tikzpicture}
}
\caption{Idle time function of \textsc{A}}
\label{subfig:idle_unfair}
\end{subfigure}
\hfill
\begin{subfigure}[b]{0.45\textwidth}
\centering
\resizebox{\linewidth}{!}{
\begin{tikzpicture}
    
    % Draw payoff axis.
    \draw[->, thin] (0,0)--(8.33,0);
    \draw[->, thin] (0,0)--(0, 2.5) node[above]{$L - \mathcal{I}(t)$};

    % Draw axis labels.
    \draw[thin] (0, -0.1) node[below]{$0$}--(0, 0.1);
    \draw[thin] (2, -0.1) node[below]{$\frac{T}{4}$}--(2, 0.1);
    \draw[thin] (4, -0.1) node[below]{$\frac{T}{2}$}--(4, 0.1);
    \draw[thin] (6, -0.1) node[below]{$\frac{3T}{4}$}--(6, 0.1);
    \draw[thin] (8, -0.1) node[below]{$T$}--(8, 0.1);

    \draw[thin] (-0.1, 0) node[left]{$0$}--(0.1, 0);
    \draw[thin] (-0.1, 1) node[left]{$1$}--(0.1, 1);
    \draw[thin] (-0.1, 2) node[left]{$2$}--(0.1, 2);

    % Draw slack function.
    \draw[dotted, thin] (0.33, 0) -- (0.33, 1);
    \draw[dotted, thin] (1.67, 0) -- (1.67, 2);
    \draw[dotted, thin] (2.33, 1) -- (2.33, 2);
    \draw[dotted, thin] (4.33, 0) -- (4.33, 1);
    \draw[dotted, thin] (5.67, 0) -- (5.67, 1);
    \draw[dotted, thin] (6.33, 0) -- (6.33, 1);
    \draw[dotted, thin] (7.67, 0) -- (7.67, 1);
    
    \draw[thick] (0, 1) -- (0.33, 1);
    \draw[fill=black] (0.33, 1) circle (2pt);
    \zeroslack{0.33}{1.67};
    \slack{1.67}{2.33}{2};
    \slack{2.33}{4.33}{1};
    \draw[fill=white] (2.33, 1) circle (2pt);
    \zeroslack{4.33}{5.67};
    \slack{5.67}{6.33}{1};
    \zeroslack{6.33}{7.67};
    \draw[thick] (7.67, 1) -- (8.33, 1);
    \draw[fill=black] (7.67, 1) circle (2pt);

    % Draw clusters.
    \draw[<->] (1.67, 2.5) -- node[midway, above] {$S_1$} (4.33, 2.5);
    \draw[<->] (5.67, 1.5) -- node[midway, above] {$S_2$} (6.33, 1.5);
    \draw[<-] (7.67, 1.5) -- node[midway, above] {$S_3$} (8.33, 1.5);
    \draw[->] (0, 1.5) -- (0.33, 1.5);

\end{tikzpicture}
}
\caption{Idle time function of \textsc{B}}
\label{subfig:idle_fair}
\end{subfigure}
\begin{subfigure}[b]{0.45\textwidth}
\centering
\begin{tikzpicture}
    % Clusters.
    \node[draw=black, shape=circle] at (0, 2) (1) {$S_1$};
    \node[draw=black, shape=circle] at (3, 2) (2) {$S_2$};
    \node[draw=black, shape=circle] at (0, 0) (3) {$S_3$};
    \node[draw=black, shape=circle] at (3, 0) (4) {$S_4$};
    
    % Connections.
    \draw[->, bend right] (1) to node[midway, left] {$I_2$} (3);
    \draw[->, bend right] (3) to node[midway, right] {$I_1$} (1);
    \draw[->, bend right] (2) to node[midway, left] {$I_3$} (4);
    \draw[->, bend right] (4) to node[midway, right] {$I_4$} (2);
\end{tikzpicture}
\caption{Idle interval graph of \textsc{A}}
\label{subfig:connectivity_unfair}
\end{subfigure}
\hfill
\begin{subfigure}[b]{0.45\textwidth}
\centering
\begin{tikzpicture}
    % Clusters.
    \node[draw=black, shape=circle] at (0, 2) (1) {$S_1$};
    \node[draw=black, shape=circle] at (3, 2) (2) {$S_2$};
    \node[draw=black, shape=circle] at (0, 0) (3) {$S_3$};
    
    % Connections.
    \draw[->, bend right] (2) to node[midway, above] {$I_1$} (1);
    \draw[->, bend right] (1) to node[midway, below] {$I_2$} (2);
    \draw[->, bend right] (1) to node[midway, left] {$I_3$} (3);
    \draw[->, bend right] (3) to node[midway, right] {$I_4$} (1);
\end{tikzpicture}
\caption{Idle interval graph of \textsc{B}}
\label{subfig:connectivity_fair}
\end{subfigure}
\caption{Idle time function and idle interval graph of instance \textsc{A} and \textsc{B}.}
\label{fig:fair_idle}
\end{figure}
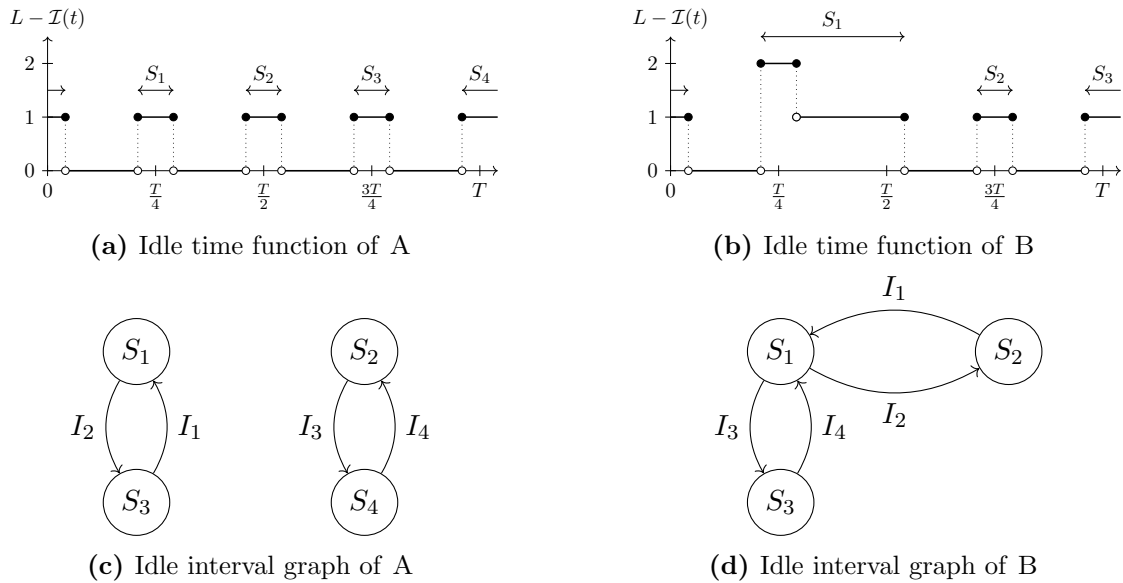

By optimality, efficient assignments do not make use of transition arcs outside idle intervals. 
\begin{lemma}
In any optimal periodic assignment, every transition arc is contained in an idle interval.
\end{lemma}

\begin{proof}
Denote the optimal assignment by $\mathcal{M}^*$ and suppose, to the contrary, that transition arc $(i, j) \in \mathcal{M}^*$ is not contained in any idle interval. Choose any $t$ in the $T$-periodic interval $[b_i, a_j]$. Since $(i, j)$ is not contained in any idle interval, it follows from \autoref{def:idle_interval} that $\mathcal{I}(t) = L$. As arc $(i, j)$ is also active at time $t$, it holds that $\mathcal{I}(t) + \mathcal{M}^*(t) \geq L + 1$. By statement $(iii)$ of \autoref{thm:pap}, this contradicts optimality of $\mathcal{M}^*$.
\end{proof}

Our goal is to study how the connectivity of idle intervals determines the number of disjoint cycles in a periodic assignment. To this end, we introduce the idle interval graph, representing how the various tasks connect different idle intervals.
\begin{definition} \label{def:idle_graph}
Let $\mathcal{S}$ denote the set of idle intervals. The idle interval graph is the directed multigraph $\mathcal{H} = (\mathcal{S}, \mathcal{B})$ that contains an arc $(k, l)$ for every task $i \in \mathcal{I}$ satisfying $a_i \in [s_k, e_k]$ and $b_i \in [s_l, e_l]$.
\end{definition}
To see that this is well-defined, note that $\mathcal{I}(a_i), \mathcal{I}(b_i) < L$ for all $i \in \mathcal{I}$. In other words, for each task there are unique idle intervals at its start and end time, respectively. By construction, the number of arcs $|\mathcal{B}|$ is always equal to exactly $n$. Moreover, any feasible periodic assignment covers all tasks once, and thereby implicitly includes all arcs of $\mathcal{H}$.

Figures~\ref{subfig:connectivity_unfair} and~\ref{subfig:connectivity_fair} illustrate the idle interval graphs of instances \textsc{A} and \textsc{B}, respectively. The number of idle intervals, and hence the number of nodes in the graph, differs across the two instances. Moreover, we find that the idle interval graph of instance \textsc{B} is connected, while that of instance \textsc{A} is not. 

It turns out that efficient solutions on instances with a connected idle interval graph always satisfy the following condition: if they are not fair, they contain overlapping transition arcs belonging to disjoint cycles. 

\begin{lemma} \label{lem:overlap}
If the idle interval graph is weakly connected, and an optimal periodic assignment $\mathcal{M}^*$ consists of disjoint cycles, then there exists a pair of overlapping transition arcs belonging to disjoint cycles.
\end{lemma}

\begin{proof}
We first show that at least two disjoint cycles must contain transition arcs traversing the same idle interval. To the contrary, suppose that no two cycles share arcs in an idle interval. Take any cycle $\mathcal{C}_{1} \in \mathcal{M}^*$, and denote by $\mathcal{S}_1 \subseteq \mathcal{S}$ all the idle intervals traversed along this cycle. By assumption, all the cycles in $\mathcal{M}' = \mathcal{M}^* \setminus \mathcal{C}_{1}$ only traverse idle intervals in $\mathcal{S}' = \mathcal{S} \setminus \mathcal{S}_1$. Clearly, $\mathcal{S}'$ is nonempty. Let $\mathcal{B}' \subseteq \mathcal{B}$ be the arcs connecting $\mathcal{S}'$ and $\mathcal{S}_1$. By connectivity of $\mathcal{H}$, this set is nonempty. Consider any arc $(k, l) \in \mathcal{B}'$. Recall that this arc is induced by some task $i \in \mathcal{I}$. Without loss of generality, assume that $k \in \mathcal{S}_1$ and $l \in \mathcal{S}'$. If task $i$ is covered by $\mathcal{C}_1$, this cycle would contain a transition arc in $l$, a contradiction. Similarly, if task $i$ is covered by some cycle in $\mathcal{M}'$, it would contain a transition arc in $k$, another contradiction. Since the task must be covered on exactly one cycle, we conclude that there must exist two disjoint cycles containing transition arcs in the same idle interval. 

Now assume that two disjoint cycles contain transition arcs in the same idle interval $[s, e] \in \mathcal{S}$. We show that their transition arcs must overlap at some time in $[s, e]$.

Let $\mathcal{C}_{1}$ be the cycle whose earliest transition arc starts at time $s$. By definition of the idle interval and optimality of $\mathcal{M}^*$, such a cycle exists. We distinguish two cases. First, suppose that the transition arcs in $\mathcal{C}_1$ are continuously active from time $s$ until time $e$. Let $\mathcal{C}_2$ be any other disjoint cycle active in the same idle interval. Clearly, any transition arc from $\mathcal{C}_2$ will overlap with some arc of $\mathcal{C}_1$. Alternatively, suppose that the transition arcs in $\mathcal{C}_1$ are active from time $s$ until some time $t_1 < e$, and potentially later in the idle interval too. Let $\mathcal{C}_2$ be the disjoint cycle passing through the same idle interval whose earliest transition arc has the second-earliest start time. Denote this start time by $t_2$. If $t_2 \leq t_1$, the corresponding transition arc must intersect with some arc in $\mathcal{C}_1$. If $t_2 > t_1$, no transition arcs are active in the interval $(t_1, t_2)$. This contradicts the definition of the idle interval and optimality of $\mathcal{M}^*$, stating that $\mathcal{M}^*(t) \geq 1$ for all $t \in [s, e]$. Hence, at least two transition arcs from disjoint cycles must overlap.
\end{proof}

As outlined before, these overlapping arcs provide opportunities for patching. Indeed, we show that patching can always be used to obtain a fair and efficient solution on instances satisfying the conditions of \autoref{lem:overlap}.

\begin{theorem} \label{thm:fair_instance}
An instance admits a fair solution with $L$ workers if and only if the idle interval graph is weakly connected. 
\end{theorem}
\begin{proof}
First, let $\mathcal{M}^*$ be any fair solution with $L$ workers. This single cycle visits all nodes (idle intervals) and arcs (tasks) of the idle interval graph. Clearly, the idle interval graph must be connected.

Now suppose that the idle interval graph is connected, and let $\mathcal{M}^*$ be any solution with $L$ workers. If it consists of a single cycle, we are done. Assume that it consists of at least two disjoint cycles. By \autoref{lem:overlap}, there exists a pair of overlapping transition arcs $(i_1, j_1), (i_2, j_2)$ belonging to disjoint cycles. We can reduce the number of disjoint cycles by one through a patching operation. In particular, we replace the original, overlapping transition arcs by $(i_1, j_2)$ and $(i_2, j_1)$. The total transition time of the assignment, and hence the number of required workers, is unaffected by this operation. The number of disjoint cycles, however, decreases by one. As the original number of disjoint cycles is at most $n$, repeating this procedure at most $n-1$ times is guaranteed to return a fair solution with $L$ workers. 
\end{proof}

To illustrate the patching idea, we return to the idle interval graphs of instances \textsc{A} and \textsc{B} in Figures~\ref{subfig:connectivity_unfair} and~\ref{subfig:connectivity_fair}, respectively. In line with \autoref{thm:fair_instance}, we find that the graph of instance \textsc{A}, for which no fair solution with $L=2$ workers exists, is not connected. In contrast, the graph of instance \textsc{B} is connected, showing that this instance admits a fair solution with two workers. It does not imply, however, that every efficient solution to this instance is fair. For example, the solution consisting of two disjoint cycles $(I_1 \rightarrow I_2 \rightarrow I_1)$ and $(I_3 \rightarrow I_4 \rightarrow I_3)$ is efficient but not fair. However, these two disjoint cycles can be patched to form a single fair solution with two workers. For example, transition arcs $(I_1, I_2)$ and $(I_4, I_3)$ overlap at time $t = \frac{T}{4}$. Patching these two arcs yields the fair and efficient solution $(I_1 \rightarrow I_3 \rightarrow I_4 \rightarrow I_2 \rightarrow I_1)$.

We conclude our theoretical analysis by pointing out that \autoref{thm:fair_instance} provides an elegant, alternative way of arriving at the price of fairness in \autoref{col:pof}. In particular, we can view the price of fairness as the minimum increase in the load required to make the idle interval graph connected. Consider an instance for which the idle interval graph is unconnected. It can easily be made connected by adding a series of artificial tasks that span the period exactly once and start and end at the start and end times of consecutive intervals, respectively. The resulting instance has load $L + 1$, and always admits a fair assignment using $L + 1$ workers.

\subsection{Patching}
\label{subsec:fair_algo} 

The proof of \autoref{thm:fair_instance} shows that, on instances admitting a fair and efficient solution, any efficient solution containing disjoint cycles can be made fair by a finite number of patching operations. We now show how to efficiently implement such a patching procedure and obtain a $\mathcal{O}(n \log n)$ exact algorithm for FPAP, which we call \textsc{Patching}.

\begin{algorithm}
\caption{\textsc{Patching}}
\label{alg:patching}
\begin{algorithmic}[1]
\STATE Compute efficient assignment $\mathcal{M} \leftarrow$ \textsc{Shift-Sort-and-Match} 
\STATE Compute number of disjoint cycles $C$ in $\mathcal{M}$
\STATE Initialize array $U$, where entry $U[i]$ stores original cycle index of task $i$
\STATE Initialize array $V$ of size $C$, where $V[k]$ stores index of cycle that cycle with original index $k$ has been patched to. Initially, $V[k] = k$ for all $k \in [C]$
\STATE Construct array of transition arcs $R$, sorted by increasing start time
\STATE Initialize patching arc $(i, j) \leftarrow R[1]$
\FOR{$m = 2, \dots, |R|$}
    \STATE Retrieve current arc $(k ,l) \leftarrow R[m]$
    \IF{$b_k > a_j$}
         \STATE Update patching arc $(i, j) \leftarrow (k, l)$
    \ELSIF{$V[U[k]] \neq V[U[i]]$}
        \STATE Perform patching by replacing arcs $(i,j), (k, l)$ in $\mathcal{M}$ with $(i, l), (k, j)$
        \STATE Update cycle indices $V[U[k]], V[U[l]] \leftarrow V[U[i]]$
        \STATE $C \leftarrow C - 1$
        \IF{$a_j > a_l$}
            \STATE Update patching arc $(i, j) \leftarrow (k, j)$
        \ELSE
            \STATE Update patching arc $(i, j) \leftarrow (i, l)$
        \ENDIF
    \ELSIF{$a_l > a_j$}
        \STATE Update patching arc $(i, j) \leftarrow (k, l)$
    \ENDIF 
\ENDFOR 
\IF{$C = 1$}
    \STATE \textbf{return} $\mathcal{M}$
\ENDIF
\STATE \textbf{return} \textsc{Nearest Neighbor}
\end{algorithmic}
\end{algorithm}

The algorithm is described in Algorithm~\ref{alg:patching}. It effectively performs the patching procedure outlined in the proof of \autoref{thm:fair_instance}. Starting from an efficient periodic assignment, it processes all transition arcs in chronological order, patching overlapping arcs that belong to disjoint cycles. We store the cycle indices in a dedicated two-array data structure, to ensure that the cycle indices can be updated in constant time after each patching operation. Once all arcs have been processed, the assignment is guaranteed to be free from overlapping transition arcs belonging to disjoint cycles. By \autoref{thm:fair_instance}, this returns a fair and efficient solution on all instances whose idle interval graph is connected. If the assignment still consists of disjoint cycles, we conclude that the idle interval graph must be disconnected. In this case, at least $L + 1$ workers are required in a fair solution, and \textsc{Nearest Neighbor} is called to find such a solution.

To ensure that overlapping transition arcs of disjoint cycles can always be patched, \textsc{Patching} makes use of a so-called \textit{patching arc}. This patching arc is the transition arc with the maximum end time among all previously processed arcs, including arcs obtained through patching. Since arcs are processed in increasing order of start time, the next processed transition arc will always overlap with the patching arc, provided that it belongs to the same idle interval. This way, any transition arc belonging to a different cycle can be successfully patched. It follows that \textsc{Patching} correctly eliminates all overlapping transition arcs belonging to disjoint cycles.

\begin{theorem}
The \textsc{Patching} algorithm returns an optimal fair periodic assignment in $\mathcal{O}(n \log n)$ time.
\end{theorem}
\begin{proof}
We start with a complexity analysis. Computing an efficient assignment with \textsc{Shift-Sort-and-Match} takes $\mathcal{O}(n \log n)$. One can compute all cycles in $\mathcal{O}(n)$ by iterating once over all transition arcs. Constructing the cycle index arrays $U$ and $V$ also takes linear time. Sorting the transition arcs $R$ by increasing start time requires $\mathcal{O}(n \log n)$. All $\mathcal{ON}(n)$ operations in the for-loop take $\mathcal{O}(1)$: patching itself takes constant time, and the arrays $U$ and $V$ allow us to update the cycle index of each task in constant time as well. Finally, calling \textsc{Nearest Neighbor} takes $\mathcal{O}(n \log n)$. The overall time complexity becomes $\mathcal{O}(n \log n)$.

It is clear that \textsc{Patching} returns a fair periodic assignment. To prove optimality, we distinguish between two cases. If the instance admits a fair solution with $L$ workers, by the proof of \autoref{thm:fair_instance} it suffices to show that \textsc{Patching} successfully patches all overlapping transition arcs belonging to different cycles. If the instance does not admit such a solution, \textsc{Nearest Neighbor} returns an optimal fair solution requiring $L + 1$ workers.

We now show that the algorithm correctly eliminates  all overlapping transition arcs belonging to different cycles. In particular, we show that patching arc $(i, j)$ always overlaps with, and hence can be patched with, the next processed transition arc $(k, l)$, whenever $(k, l)$ belongs to the same idle interval as $(i, j)$. 

Without loss of generality, assume that the initial patching arc $R[1]$ marks the start of an idle interval. Assume that the current patching arc is $(i, j)$, and we are processing arc $(k, l)$. In case the \textsc{IF}-statement on Line 9 evaluates to true, it must hold that $(k, l)$ belongs to a different idle interval than $(i, j)$. To the contrary, suppose that they belong to the same idle interval but $b_k > a_j$. Since we always update the patching arc to the processed transition arc with the latest end time, this would imply that no transition arcs are active in the interval $(a_j, b_k)$, contradicting the definition of idle interval and efficiency of $\mathcal{M}$. Since the arcs belong to different idle intervals, we do not perform patching but simply update the patching arc. 

Now, assume that the two arcs belong to the same idle interval but to disjoint cycles, i.e., Line 11 evaluates to true. We can patch the two arcs whenever they overlap, i.e., whenever the interval $[b_i, a_j] \cap [b_k, a_l] = [\max \{b_i, b_k\}, \min \{ a_j, a_l\}]$ is nonempty. Since we process transition arcs in increasing order of their start time, it holds that $b_i \leq b_k$. From Line 9, it follows that $a_j \geq b_k$. Hence, the two arcs overlap, and we perform a patching operation. Moreover, we update the patching arc to the processed transition arc with the latest end time. 

Finally, in case the two arcs belong to the same idle interval but the same cycle, we do not perform a patching operation. Instead, on Line 21 we update the patching arc to preserve the required property that it has the maximum end time among all processed transition arcs. 
\end{proof}

Interestingly, we are not aware of a more efficient algorithm for testing the \textit{existence} of a fair and efficient periodic assignment than \textsc{Patching}, which directly computes the optimal fair assignment. The obvious alternative way of testing the conditions of \autoref{thm:fair_instance} is to construct the idle interval graph and perform a depth-first search to determine its connectivity. While the latter step takes linear time as $\mathcal{O}(|\mathcal{S}| + |\mathcal{B}|) = \mathcal{O}(n)$, computing the idle intervals themselves requires a sorting of the tasks, bringing the time complexity to $\mathcal{O}(n \log n)$, i.e., the same as \textsc{Patching}.

\section{Fair versus Balanced Assignments}

Thus far, we considered assignments that define periodic schedules: if there are $q$ workers, a worker performing some task in period $p$ also performs this task in periods $p+zq$ for any $z\in \mathbb{Z}_{\geq 0}$. In contrast, \cite{gachet_et_al:OASIcs.ATMOS.2024.5} and \cite{gachet2025balancedassignmentsperiodictasks} consider general \textit{aperiodic} assignments and develop conditions for when such an assignment is \textit{balanced}. Informally, this means that an assignment is fair in the long term average. In this section, we show that there are no benefits for allowing aperiodic assignments, or periodic assignments that repeat with a longer period than the number of workers: if there is a balanced, potentially aperiodic assignment, there also exists a fair periodic assignment with the same number of workers. 

Following \cite{gachet_et_al:OASIcs.ATMOS.2024.5}, an assignment for $q$ workers is a function $f: \tasks \times \mathbb{Z}_{\geq 0} \rightarrow [q]$, where $f(i,r)=m$ means that the $r'$th occurrence of task $i$ is performed by worker $m$. An assignment is feasible if no worker is assigned to two overlapping tasks. More formally, this requires that $f(i, r) \neq f(i', r')$ whenever $(a_i + r, b_i) \cap (a_{i'} + r', b_{i'} + r') \neq \emptyset$ for all $i \neq i'$ and $r, r'$. An assignment is balanced if for all tasks $i\in\tasks$ and all workers $m$

$$\lim_{p\rightarrow \infty} \frac{1}{p}\left| \{ r\in [p] : f(i,r)=m\}\right| = \frac{1}{q}.$$
In other words, in the long term average, all workers perform all tasks with the same proportion.

We let $\tasks^r$ denote the roll-out of the task set, i.e., the set containing the intervals $(a_i+rT,b_i+rT)$ for all $i\in \tasks$ and $r\in \mathbb{Z}_{\geq 0}$. Let $\tasks^r(t)$ denote the number of active tasks at time $t\geq 0$.
The definition of idle intervals naturally extends to the roll-out:
\begin{definition} \label{def:idle_interval2}
A rolled-out idle interval $k$ is a maximal closed interval $[s_k, e_k]$ satisfying $L - \tasks^r(t) > 0$ for all $t \in [s_k, e_k]$.
\end{definition}
Since tasks repeat periodically, it holds that $\tasks^r(t+rT)=\tasks(t)$ for all $r\in \mathbb{Z}_{\geq 0}$ and all $t$. Hence, every regular idle interval (as defined in Section~\ref{sec:fair}) is associated with an infinite number of corresponding rolled-out idle intervals. Analogous to the periodic case, we can define a graph that represents the connections between the idle intervals:
\begin{definition} \label{def:idle_graph_rolled}
Let $\mathcal{S}^r$ denote the set of idle intervals. The rolled-out idle interval graph is the directed multigraph $\mathcal{H}^r = (\mathcal{S}^r, \mathcal{B}^r)$ that contains an arc $(k, l)$ for every task $i \in \tasks^r$ satisfying $a_i \in [s_k, e_k]$ and $b_i \in [s_l, e_l]$.
\end{definition}
Since there is an infinite number of rolled-out idle intervals, the rolled-out idle interval graph is an infinite graph. Moreover, while the periodic idle interval graph is cyclic, the rolled-out idle interval graph is acyclic, as all arcs go forward in time. \autoref{fig:rolled_out} shows the first three periods of the rolled-out idle interval graph corresponding to a roll-out of instance \textsc{B}.

\newcommand{\arc}[4]{
    % #1: x-coordinate, #2: y-coordinate, #3: radius
    \pgfmathsetmacro{\innerR}{#3-0.1}
    \pgfmathsetmacro{\outerR}{#3+0.1}
    \pgfmathsetmacro{\middle}{0.1 * #1 + 0.9 * #2}
    \draw[black] ({#3*cos(#1)}, {#3*sin(#1)}) arc (#1:#2:#3);
    \draw[black] ({\innerR *cos(#1)}, {\innerR * sin(#1)}) -- ({\outerR*cos(#1)}, {\outerR*sin(#1)});
    \draw[black] ({\innerR *cos(#2)}, {\innerR*sin(#2)}) -- ({\outerR*cos(#2)}, {\outerR*sin(#2)});
    \node at ({(\outerR + 0.11) *cos(\middle)}, {(\outerR + 0.11) *sin(\middle)}) {#4};
}

\tikzset{->-/.style={decoration={
			markings,
			mark=at position #1 with {\arrow{>}}},postaction={decorate}}}

\newcommand{\slack}[3]{
    % #1: x-start, #2: x-end #3: y
    \draw[thick] (#1, #3) -- (#2, #3);
    \draw[fill=black] (#1, #3) circle (2pt);
    \draw[fill=black] (#2, #3) circle (2pt);
}

\newcommand{\zeroslack}[2]{
    % #1: x-start, #2: x-end 
    \draw[thick] (#1, 0) -- (#2, 0);
    \draw[fill=white] (#1, 0) circle (2pt);
    \draw[fill=white] (#2, 0) circle (2pt);
}

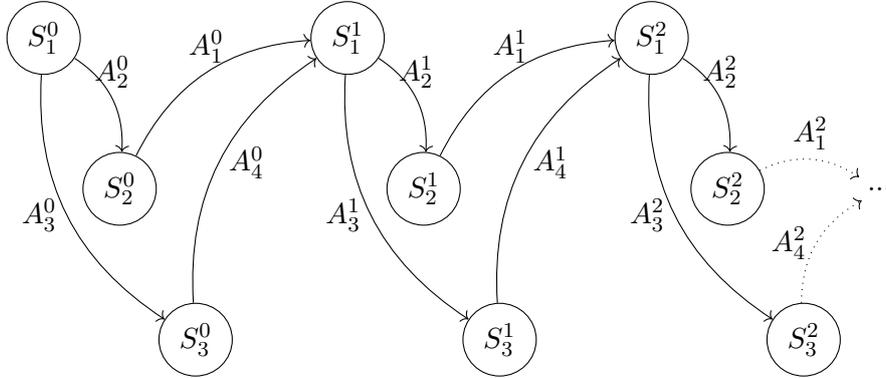
\begin{figure}[htbp!]
\centering
\begin{tikzpicture}
    % r=1
    \node[draw=black, shape=circle] at (0, 4) (1) {$S^0_1$};
    \node[draw=black, shape=circle] at (1, 2) (2) {$S^0_2$};
    \node[draw=black, shape=circle] at (2, 0) (3) {$S^0_3$};

    % r=2
    \node[draw=black, shape=circle] at (4, 4) (4) {$S^1_1$};
    \node[draw=black, shape=circle] at (5, 2) (5) {$S^1_2$};
    \node[draw=black, shape=circle] at (6, 0) (6) {$S^1_3$};

     % r=3
    \node[draw=black, shape=circle] at (8, 4) (7) {$S^2_1$};
    \node[draw=black, shape=circle] at (9, 2) (8) {$S^2_2$};
    \node[draw=black, shape=circle] at (10, 0) (9) {$S^2_3$};

    % Infinite node.
    \node at (11, 2) (10) {...};
    
    % Connections.
    \draw[->, bend left] (1) to node[midway, above] {$A^0_2$} (2);
    \draw[->, bend right] (1) to node[midway, left] {$A^0_3$} (3);

    \draw[->, bend left] (4) to node[midway, above] {$A^1_2$} (5);
    \draw[->, bend right] (4) to node[midway, left] {$A^1_3$} (6);

    \draw[->, bend left] (7) to node[midway, above] {$A^2_2$} (8);
    \draw[->, bend right] (7) to node[midway, left] {$A^2_3$} (9);
    
    \draw[->, bend left] (3) to node[midway, right] {$A^0_4$} (4);
    \draw[->, bend left] (6) to node[midway, right] {$A^1_4$} (7);
     
    \draw[->, bend left] (2) to node[midway, above] {$A^0_1$} (4);
    \draw[->, bend left] (5) to node[midway, above] {$A^1_1$} (7);

    \draw[->, bend left, dotted] (8) to node[midway, above] {$A^2_1$} (10);
    \draw[->, bend left, dotted] (9) to node[midway, left] {$A^2_4$} (10);
    
\end{tikzpicture}
\caption{Rolled-out idle interval graph of instance \textsc{B}, showing the first three periods.}
\label{fig:rolled_out}
\end{figure}

There is a many-to-one mapping from nodes and arcs in the rolled-out graph to nodes and arcs in the periodic graph. This results in the following observation:

\begin{observation}
\label{obs:connection}
   $\mathcal{H}^r$ is weakly connected if and only $\mathcal{H}$ is connected.
\end{observation}

We can now prove a counterpart of \autoref{thm:fair_instance} for aperiodic instances.

\begin{theorem} \label{thm:balanced_instance}
An instance admits a balanced assignment with $L$ workers if and only if $\mathcal{H}^r$ is weakly connected.
\end{theorem}
\begin{proof}
    If $\mathcal{H}^r$ is weakly connected, $\mathcal{H}$ is connected. It follows from Theorem~\ref{thm:fair_instance} that there exists a fair periodic assignment with $L$ workers, which also defines a balanced assignment. 
    
    To prove the other direction, assume that the rolled-out idle interval graph is not weakly connected. This implies that there are tasks that are not connected through idle intervals. In a balanced schedule, however, workers must perform all (periodic) tasks, necessitating the use of long transition arcs that cross idle intervals. At such instants, there are $L$ active task arcs and at least one active transition arc, so at least $L+1$ workers are required in a balanced solution. 
\end{proof}
Recall that there always exists a fair periodic assignment with $L+1$ workers. It directly follows from \autoref{thm:balanced_instance} that imposing balancedness instead of the stricter fairness criterion does not provide any benefits in terms of efficiency:
\begin{corollary} 
An instance admits a balanced assignment with $q$ workers if and only if the instance admits a fair periodic assignment with $q$ workers.
\end{corollary}

\bibliography{references}

\end{document}